\documentclass{llncs}
\usepackage{color}
\usepackage{mathtools}

\usepackage{graphicx}
\usepackage{xspace}
\usepackage{amsmath}
\usepackage{amsfonts}
\usepackage{enumerate}

\usepackage{amssymb}
\usepackage{mathrsfs}
\usepackage{color}
\usepackage{url}
\usepackage{etoolbox}
\usepackage{pgfplots}
\usepackage{pgfplots}

\usepackage{listings}
\usepackage{xcolor}
\newcounter{instr}

\newcommand{\defAs}
{=}

\newlength{\gnat}

\newlength{\gnatb}

\newcommand*{\QEDA}{\hfill\ensuremath{\blacksquare}}%

\title{Efficient Online String Matching Based on Characters Distance Text Sampling\thanks{We gratefully acknowledge support from ``Universit\`{a} degli Studi di Catania, Piano della Ricerca 2016/2018 Linea di intervento 2".}}

\author{Simone Faro$^{\dagger}$, Arianna Pavone$^{\ddagger}$ and Francesco Pio Marino$^{\dagger}$}
\institute{
$^{\dagger}$Diparitmento di Matematica e Informatica,\\
		Universit\`a di Catania, Viale A.Doria n.6, 95125 Catania, Italy\\
		\email{faro@dmi.unict.it}\\[0.2cm]
$^{\ddagger}$Dipartimento di Scienze Cognitive,\\
             Universit\`a di Messina, via Concezione n.6, 98122, Messina, Italia \\
             \email{pavone@unime.it}
}

\begin{document}

\maketitle

\begin{abstract}
Searching for all occurrences of a pattern in a text is a fundamental problem in computer science with applications in many other
fields, like natural language processing, information retrieval and computational biology. \emph{Sampled string matching} is an efficient approach recently introduced in order to overcome the prohibitive space requirements of an index construction, on the one hand, and drastically reduce searching time for the online solutions, on the other hand.

In this paper we present a new algorithm for the sampled string matching problem, based on a characters distance sampling approach.  The main idea is to sample the distances between consecutive occurrences of a given \emph{pivot} character and then to search online the sampled data for any occurrence of the sampled pattern, before verifying the original text. 

From a theoretical point of view we prove that, under suitable conditions, our solution can achieve both linear worst-case time complexity and optimal average-time complexity.
From a practical point of view it turns out that our solution shows a sub-linear behaviour in practice and speeds up online searching by a factor of up to $9$, using limited additional space whose amount goes from $11\%$ to $2.8\%$ of the text size, with a gain up to $50\%$ if compared with previous solutions.\\[0.2cm]
\textbf{keywords:} string matching, text processing, efficient searching, text indexing

\end{abstract}

\section{Introduction}

\emph{String matching} is a fundamental problem in computer science and in the wide domain of text processing.  It consists in finding all the occurrences of a given pattern $x$, of length $m$, in a large text $y$, of length $n$, where both sequences are composed by characters drawn from an alphabet $\Sigma$ of size $\sigma$.
%
Although data are memorized in different ways, textual data remains the main form to store information. This is particularly evident in literature and in linguistic where data are in the form of huge corpus and dictionaries. But this apply as well to computer science where large amount of data are stored in linear files. And this is also the case, for instance, in molecular biology where biological molecules are often approximated as sequences of nucleotides or amino acids. Thus the need for more and more faster solutions to text searching problems.

Applications require two kinds of solutions: \emph{online} and \emph{offline} string matching. Solutions based on the first approach assume that the text is not preprocessed and thus they need to scan the text \emph{online}, when searching. Their worst case time complexity is $\Theta(n)$, and was achieved for the first time by the well known Knuth-Morris-Pratt (KMP) algorithm \cite{Ref1}, while the optimal average time complexity of the problem is $\Theta(n \log_{\sigma} m/m)$ \cite{Ref2}, achieved for example by the Backward-Dawg-Matching (BDM) algorithm \cite{Ref3}. 
Many string matching algorithms have been also developed to obtain sub-linear performance in practice \cite{Ref34}. Among them the Boyer-Moore-Horspool algorithm \cite{Ref36,Ref4} deserves a special mention, since it has been particularly successful and has inspired much work.

Memory requirements of this class of algorithms are very low and generally limited to a precomputed table of size $O(m\sigma)$ or $O(\sigma^2)$ \cite{Ref34}. However their performances may stay poor in many practical cases, especially when used for processing huge input texts and short patterns.\footnote{Search speed of an online string matching algorithm may depend on the length of the pattern. Typical search speed of a  fast solution, on a modern laptop computer, goes from $1$ GB/s (in the case of short patterns) to $5$ GB/s (in the case of very long patterns) \cite{Ref35}.}

Solutions based on the second approach tries to drastically speed up searching by preprocessing the text and building a data structure that allows searching in time proportional to the length of the pattern. For this reason such kind of problem is known as \emph{indexed searching}. Among the most efficient solutions to such problem we mention those based on suffix trees \cite{Ref5}, which find all occurrences in $O(m+occ)$-worst case time, those based on suffix arrays \cite{Ref6}, which solve the problem in $O(m + \log n + occ)$ \cite{Ref6}, where $occ$ is the number of occurrences of $x$ in $y$, and those based on the FM-index \cite{FM05} (Full-text index in Minute space), which is a compressed full-text substring index based on the Burrows-Wheeler transform allowing compression of the input text while still permitting fast substring queries.
However, despite their optimal time performances\footnote{Search speed of a fast offline solution do not depend on the length of the text and is typically under $1$ millisecond per query.},  space requirements of such data structures are from $4$ to $20$ times the size of the text, which is too large for many practical applications. 

A different solution to the problem is to compress the input text and search online directly the compressed data in order to speed-up the searching process using reduced extra space. Such problem, known in literature as \emph{compressed string matching},  has been widely investigated in the last few years. Although efficient solutions exist for searching on standard compressions schemes, as Ziv-Lempel \cite{Ref32} and Huffman \cite{Ref33}, the best practical behaviour are achieved by ad-hoc schemes designed for allowing fast searching \cite{Ref26,Ref27,Ref28,Ref29,Ref30}. These latter solutions use less than 70\% of text size extra space (achieving a compression rate over 30\%) and are twice as fast in searching as standard online string matching algorithms.  A drawback of such solutions is that most of them still require significant implementation efforts and a high time for each reported occurrence.

A more suitable solution to the problem is \emph{sampled string matching}, recently introduced by Claude \emph{et al.} \cite{Ref37}, which consists in the construction of a succint sampled version of the text and in the application of any online string matching algorithm directly on the sampled sequence. The drawback of this approach is that any occurrence reported in the sampled-text may require to be verified in the original text. However a sampled-text approach my have a lot of good features:
it may be easy to implement,  may require little extra space and my allow fast searching. Additionally it my allow fast updates of the data structure.
Specifically the solution of Claude \emph{et al.} is based on an alphabet reduction. Their algorithm has an extra space requirement which is only 14\% of text size and is up to 5 times faster than standard online string matching on English text. Thus it turns out to be one of the most effective and flexible solution for this kind of searching problems.

\subsection{Our contribution and organization of the paper}
In this paper we present a new approach to the sampled string matching problem based on alphabet reduction and characters distance sampling. Instead of sampling characters of the text belonging to a restricted alphabet, we divide the text in blocks of size $k$ and sample positions of such characters inside the blocks. The sampled data is then used to filter candidate occurrences of the pattern, before verifying the whole match in the original text.

Our new approach is simple to implement and guarantees approximately the same performances as the solution proposed by Claude \emph{et al.}  in practice. However it is faster in the case of short patterns and, more interesting, may require only 5\% of additional extra space. 

We prove also that if the underlying algorithm used for searching the sampled text for the sampled pattern, achieves optimal worst-case and average-case time complexities, then our new solution attains the same optimal complexities, at least for patterns with a length of (at most) few hundreds of characters.

\smallskip

The paper is organized as follows. Firstly, we present in Section \ref{sec:indexsearch} the efficient sampling solution proposed by Claude \emph{et al.} Then, in Section \ref{sec:new}, we introduce our new sampling approach, discuss its good features and present its practical behaviour. In Section \ref{sec:results} we compare the two approaches and present some experimental results. Finally, in Section \ref{sec:conclusions} we draw our conclusions and discuss some further improvements.


\section{Sampled String Matching} \label{sec:indexsearch}
The task of the \emph{sampled string matching} problem is to find all occurrences of a given pattern $x$, of length $m$, in a given text $y$, of length $n$, assuming that a fast and succint preprocessing of the text is allowed in order to build a data-structure, which is used to speed-up the searching phase. For its features we call such data structure a \emph{partial-index} of the text.

In order to be of any practical and theoretical interest a partial-index of the text should:
\begin{itemize}
\item[(1)] 
\emph{be succint}: since it must be maintained together with the original text, it should require few additional spaces to be constructed;
\item[(2)]
\emph{be fast to build}: it should be constructed using few computational resources, also in terms of time. This should allow the data structure to be easily built online when a set of queries is required;
\item[(3)]
\emph{allow fast search}: it should drastically increase the searching time of the underlying string matching algorithm. This is one of the main features required by this kind of solutions;
\item[(4)]
\emph{allow fast update}: it should be possible to easily and quickly update the data structure if modifications have been applied on the original text. A desirable update procedure should be at least as fast as the modification procedure on the original text.
\end{itemize}

In this section we briefly describe the efficient text-sampling approach proposed by Claude \emph{et al.} \cite{Ref37}. To the best of our knowledge it is the most effective and flexible solution known in literature for such a problem. We will refer to this solution as the Occurrence-Text-Sampling algorithm (\textsc{Ots}).

\subsection{The Occurrence Text Sampling algorithm}
Let $y$ be the input text, of length $n$, and let $x$ be the input pattern, of length $m$, both over an alphabet $\Sigma$ of size $\sigma$.
The main idea of their sampling approach is to select a subset of the alphabet, $\hat{\Sigma} \subset \Sigma$ (the sampled alphabet), and then to construct a partial-index as the subsequence of the text (the sampled text) $\hat{y}$, of length $\hat{n}$, containing all (and only) the characters of the sampled alphabet $\hat{\Sigma}$. More formally $\hat{y}[i] \in \hat{\Sigma}$, for all $1\leq i \leq \hat{n}$.

During the searching phase of the algorithm a sampled version of the input pattern, $\hat{x}$, of length $\hat{m}$, is constructed and searched in the sampled text. Since $\hat{y}$ contains partial information, for each candidate position $i$ returned by the search procedure on the sampled text, the algorithm has to verify the corresponding occurrence of $x$ in the original text.
For this reason a table $\rho$ is maintained in order to map, at regolar intervals, positions of the sampled text to their corresponding positions in the original text. The position mapping $\rho$ has size $\lfloor \hat{n}/q\rfloor$, where $q$ is the \emph{interval factor}, and is such that $\rho[i] = j$ if character $y[j]$ corresponds to character  $\hat{y}[q\times i]$. The value of $\rho[0]$ is set to $0$. In their paper, on the basis of an accurate experimentation, the authors suggest to use values of $q$ in the set $\{8,16,32\}$

Then, if the candidate occurrence position $j$ is stored in the mapping table, i.e if $\rho[i]=j$ for some $1\leq i\leq \lfloor \hat{n}/q\rfloor$, the algorithm directly checks the corresponding position in $y$ for the whole occurrence of $x$. Otherwise, if the sampled pattern is found in a position $r$ of $\hat{y}$, which is not mapped in $\rho$, the algorithm has to check the substring of the original text which goes from position 
$\rho[r/q] + (r \mod q) - \alpha + 1$ to position $\rho[r/q + 1] - (q - (r \mod q)) - \alpha + 1$, where $\alpha$ is the first position in $x$ such that $x[\alpha]\in\hat{\Sigma}$.

Notice that, if the input pattern does not contain characters of the sampled alphabet, i.e. id $\bar{m}=0$,  the algorithm merely reduces to search for $x$ in the original text $y$.

\begin{example}
Suppose $y =$ ``abaacabdaacabcc'' is a text of length $15$ over the alphabet $\Sigma=\{$a,b,c,d$\}$. Let $\hat{\Sigma} = \{$b,c,d$\}$ be the sampled alphabet, by omitting character  ``a''.  Thus the sampled text is $\hat{y} =$ ``bcbdcbcc''. If we map  every $q=2$ positions in the sampled text, the position mapping $\rho$ is $\langle 5, 8, 12, 14\rangle$. To search for the pattern $x =$ ``acab'' the algorithm constructs the sampled pattern $\hat{x}=$ ``cb'' and search for it in the sampled text, finding two occurrences at position 2 and 5, respectively. 
We note that $\hat{y}[2]$ is mapped and thus it suffices to verify for an occurrence starting at position 4,  finding a match. However position $\hat{y}[5]$ is not mapped, thus we have to search in the substring $y[\rho(2)+3-1 .. \rho(3)]$, finding no matches.
\end{example}

The above algorithm works well with most of the known pattern matching algorithms.  However, since the sampled patterns tend to be short, the authors implemented the search phase using the Horspool algorithm, which has been found to be fast in such setting. 

The real challenge in their algorithm is how to choose the best alphabet subset to sample. Based on some analytical results, supported by an experimental evaluation, they showed that it suffices in practice to sample the least frequent characters up to some limit.\footnote{According to their theoretical evaluation and their experimental results it turns out that, when searching on an English text,  the best performances are obtained when the least 13 characters are removed from the original alphabet.} 
Under this assumption their algorithm has an extra space requirement which is only 14\% of text size and is up to 5 times faster than standard online string matching on English texts.

For the sake of completeness it has to be noticed that in \cite{Ref37} the authors also consider indexing the sampled text. Specifically they build a suffix array indexing the sampled positions of the text, and get a sampled suffix array. This approach is similar to the sparse suffix array \cite{Ref10} as both index a subset of the suffixes, but the different sampling properties induce rather different search algorithms and performance characteristics.
More recently Grabowsky and Raniszewski \cite{GR15} proposed a more convenient indexing suffix sampling approach, with only a minimum pattern length as a requirement. 


\section{A New Algorithm Based on Characters Distance Sampling} \label{sec:new}
In this section we present a new efficient approach to the sampled string matching problem, introducing a new method for the construction of the partial-index, which turns out to require limited additional space, still maintaining the same performances of the algorithm recently introduced by Claude \emph{et al.} \cite{Ref37}.
In the next subsections we illustrate in details our idea and describe the algorithms for the construction of the sampled text and for the searching phase. 

\subsection{Characters Distance Sampling}
As above, let $y$ be the input text, of length $n$, and let $x$ be the input pattern, of length $m$, both over the alphabet $\Sigma$ of length $\sigma$. We assume that all strings can be treated as vectors starting at position $1$. Thus we refer to $x[i]$ as the $i$-th character of the string $x$, for $1\leq i \leq m$, where $m$ is the length of $x$.

We first define a sampled alphabet $\bar{\Sigma} \subset \Sigma$, which we call the set of \emph{pivot characters}. The set $\bar{\Sigma}$ could be very small, and in many practical cases could be reduced to a single character. If $|\bar{\Sigma}|=1$, the unique character of $\bar{\Sigma}$ is called the \emph{pivot} character. For simplicity we will assume in what follows that $|\bar{\Sigma}|=1$, which can be trivially generalized to the case where $|\bar{\Sigma}|>1$.

The text sampling approach used in our solution is based on the following definition of \emph{bounded position sampling}.

\begin{definition}[The Bounded Position Sampling]\label{def:pos}
Let $y$ a text of length $n$, let $C \subseteq \Sigma$ be the set of pivot characters and let $n_c$ be the number of occurrences of any character of $C$ in the input text $y$. First we define the \emph{position function}, $\delta:\{1, .., n_c\} \rightarrow \{1, .., n\}$, where $\delta(i)$ is the position of the $i$-th occurrence of any character of $C$ in $y$. Formally we have
$$
\begin{array}{rll}
	(i) & 1\leq \delta(i)<\delta(i+1)\leq n & \textrm{ for each } 1\leq i \leq n_c-1\\
	(ii) & y[\delta(i)] \in  C & \textrm{ for each } 1\leq i \leq n_c\\
	(iii) & y[\delta(i)+1..\delta(i+1)-1] \textrm{ contains no $c  \in C$ } & \textrm{ for each } 0\leq i \leq n_c\\
\end{array}
$$
where, in $(iii)$, we assume that $\delta(0)=0$ and $\delta(n_c+1) = n+1$. 

Assume now that $k$ is a given threshold constant. We define the \emph{$k$-bounded position function}, $\delta_k:\{1, .., n_c\} \rightarrow \{0, .., k-1\}$, where $\delta_k(i)$ is the position (modulus $k$) of the $i$-th occurrence of any character of $c$ in $y$. Formally we have
$$
	\delta_k(i) = \left[\delta(i)\mod k\right], \textrm{ for each } i=1, ..., n_c
$$
The \emph{$k$-bounded-position sampled version} of $y$, indicated by $\dot{y}$, is a numeric sequence, of length $n_c$ defined as
\begin{equation}\label{eq:bps}
	\dot{y} = \langle \delta_k(1), \delta_k(2), .., \delta_k(n_c) \rangle.
\end{equation}
Plainly we have $0\leq \dot{y}[i]<k$, for each $1\leq i \leq n_c$. 
\end{definition}

\begin{example}\label{ex:0}
Suppose $y =$ ``abaacbcabdada'' is a text of length 13, over the alphabet $\Sigma=\{$a,b,c,d$\}$. Let  $C = \{$``a'', ``c''$\}$ be the set of pivot characters and let $k=5$ be the threshold value.  Thus the position sampled version of $y$ is $\dot{y} = \langle1,3,4,0,2,3,1,3\rangle$. Specifically the first occurrence of a charcater in $C$ is at position 1 ($y[1]=$ a), its second occurrence is at position 3 ($y[3]=$ a). However its $5$-th occurrence is at position $7$, thus $\dot{y}[5] = [\delta(7)\mod 5] = [7\mod 5] = 2$.
\end{example}

\begin{example}\label{ex:1}
Suppose $y =$ ``abaacbcabdada'' is a text of length 13, over the alphabet $\Sigma=\{$a,b,c,d$\}$. Let  ``a''  be the pivot character and let $k=5$ be the threshold value.  Thus the position sampled version of $y$ is $\dot{y} = \langle1,3,4,3,1,3\rangle$. Specifically the first occurrence of charcater ``a'' is at position 1, its second occurrence is at position 3. However its $4$-th occurrence is at position $8$, thus $\dot{y}[4] = [\delta(4)\mod 5] = [8\mod 5] = 3$.
\end{example}

\begin{definition}[The Characters Distance Sampling]\label{def:dist}
Let $c\in\Sigma$ be the pivot character, let $n_c\leq n$ be the number of occurrences of the pivot character in the text $y$ and let $\delta$ be the position function of $y$. We define the \emph{characters distance function} $\Delta(i) = \delta(i+1)-\delta(i)$, for $1 \leq i \leq n_c-1$, as the distance between two consecutive occurrences of the character $c$ in $y$.

The \emph{characters-distance sampled version} of the text $y$ is a numeric sequence, indicated by $\bar{y}$, of length $n_c-1$ defined as
\begin{equation}\label{eq:cds}
	\bar{y} = \langle \Delta(1), \Delta(2), .., \Delta(n_c-1) \rangle.
\end{equation}
\end{definition}

Plainly we have $$\sum_{i=1}^{n_c-1}\Delta(i) \leq n-1.$$

\begin{example}
As in Example \ref{ex:1}, let $y =$ ``abaacbcabdada'' be a text of length 13, over the alphabet $\Sigma=\{$a,b,c,d$\}$. Let ``a''  be the pivot character.  Thus the character distance sampled version of $y$ is $\bar{y} = \langle2,1,4,3,2\rangle$. Specifically $\bar{y}[1]=\Delta(1)=\delta(1)-\delta(0)=3-1=2$, while $\bar{y}[3]=\Delta(4)=\delta(4)-\delta(3)=8-4=4$, and so on.
\end{example}

In order to be able to retrieve the original $i$-th position $\delta(i)$, of the pivot character, from the $i$-th element of the $k$-bounded position sampled text $\dot{y}$, we also maintain a \emph{block-mapping table} $\tau$ which stores the indexes of the last positions of the pivot character in each $k$-block of the original text, for a given input block size $k$. 

Specifically, we assume that the text $y$ is divided in $\lceil n/k \rceil$ blocks of length $k$, with the last block containing ($n\mod k$) characters. Then $\tau[i]=j$ if the $j$-th occurrence of the pivot character in $y$ is also its last occurrence in the $i$-th block. If the $i$-th block of $y$ does not contain any occurrence of the pivot character than $\tau[i]$ is set to be equal to the last position of the pivot character in one of the previous blocks. 

More formally we have, for $1\leq i \leq \lceil n/k\rceil$ 
\begin{equation}\label{tau}
	\tau[i] = \max\left(\{ j : \delta(j)\leq ik \} \cup \{0\}\right)
\end{equation}
Thus it is trivial to prove that $\tau[i]=j$ if and only if $\delta(j)\leq (ik)$ and $\delta(j+1)>(ik)$. In addition the values in the block mapping $\tau$ are stored in a non decreasing order. Formally
\begin{equation}\label{eq:order}
	\tau[i] \leq \tau[i+1],\ \forall\ 0\leq i \leq \lceil n/k \rceil
\end{equation}

\begin{example}
Let $y =$ ``caacbddcbcabbacdcadcab'' be a text of length 22, over the alphabet $\Sigma=\{$a,b,c,d$\}$. Let ``a''  be the pivot character and let $k=5$ be the block size.  The $k$-bounded position sampled version of $y$ is $\dot{y} = \langle 2, 3, 1, 4, 3, 1 \rangle$. The text is divided in $\lceil 22/5 \rceil = 5$ bloks, where the last block contains only $2$ characters. Thus the mapping table $\tau$ contains exactly $5$ entries, and specifically $\tau[1] = 2$, since the last occurrence of the pivot character in the first block corresponds to its second occurrence in $y$. Similarly, $\tau[3]=4$, $\tau[4]=5$ and $\tau[5]=6$. Observe however that, since the second block does not contain any occurrence of the pivot character we have $\tau[2] = \tau[1] = 2$.   
\end{example}

The following Lemma \ref{lemma:block_find} defines how to compute the index of the block in which the $j$-th occurennce of the pivot character is located.

\begin{lemma} \label{lemma:block_find} Let $y$ be a text of length $n$, let $c\in \Sigma$ be the pivot characters and assume $c$ occurs $n_c$ times in $y$. In addition let $\dot{y}$  be the $k$-bounded-position sampled versions of $y$, and let $\tau$ be corresponding mapping table. Then the $j$-th occurrence of the pivot character in $y$ occurs in the $i$-th block of $y$ if $\tau[i]\geq j$ and $\tau[i-1]<j$.
\end{lemma}
\begin{proof}
Assume that the $j$-th occurrence of the pivot character in $y$ occurs in the $i$-th block and let $\delta(h_1)$ be the position of the last occurrence of the pivot character in the $(i-1)$-th block. Since the character of position $\delta(j)$ occurs in the $i$-th block we have $j>h_1$. By definition we have $\tau[i-1] = h_1 < j$.

Moreover let  $\delta(h_2)$ be the position of last occurrence of the pivot character int the $i$-th block. We plainly have $\delta(h_2)\geq \delta(j)$, $j<ik$ and $h_2<ik$. Thus by equation (\ref{tau}) we have $\tau[i] \geq j$. Proving the lemma.\QEDA
\end{proof}

The following Corollary \ref{corol:position} defines the relation for computing the original position $\delta(j)$ from $\dot{y}[j]$ and the mapping table $\tau$. It trivially follows from Lemma \ref{lemma:block_find} and equation (\ref{eq:order}).

\begin{corollary} \label{corol:position} Let $y$ be a text of length $n$, let $c\in \Sigma$ be the pivot character and assume $c$ occurs $n_c$ times in $y$. In addition let $\dot{y}$  be the $k$-bounded-position sampled versions of $y$, and let $\tau$ be corresponding mapping table. Let $b = \min\{ i : \tau[i]\geq j\}$, then we have
$$
	\delta(j) = (\tau[b]-1) k + \dot{y}[j].
$$
\QEDA
\end{corollary}

Based on the relation defined by Corollary \ref{corol:position} the pseudocode shown in Figure \ref{code:getpos} (on the left) describes a procedure for computing the correct position $\delta(i)$ in $y$, from the $i$-th element of the $k$-bounded position sampled text $\dot{y}$.
In the pseudocode of \textsc{Get-Position($\tau$, $b$, $\dot{y}$, $i$)} we assume that the parameter $b$ is less or equal the actual block of the $i$-th occurrence of the pivot character, i.e. $bk \leq \delta(i)$. It returns a couple $(p, b)$, where $p=\delta(i)$ and $b$ is such that $(b-1)k < \delta(i)\leq bk$.

\begin{figure}
\begin{center}
\begin{tabular}{ll}

\begin{tabular}{ll}
	\multicolumn{2}{l}{\textsc{Get-Position}($\tau$, $b$, $\dot{y}$, $i$)}\\ 
	1. & \quad \textsf{while $\tau[b]<i$ do}\\
	2. & \quad \qquad \textsf{$b\leftarrow b+1$}\\
	3. & \quad \textsf{$p\leftarrow (b-1)\times k+\dot{y}[i]$}\\
	4. & \quad \textsf{return ($p$, $b$)}\\
	& \\
	& \\
	& \\
	& \\
\end{tabular}
&
\begin{tabular}{ll}
	\multicolumn{2}{l}{\textsc{Compute-Character-Distance-Sampling}($\dot{y}$, $\tau$)}\\ 
	1. & \quad \textsf{$\bar{y} \leftarrow \langle \rangle$}\\
	2. & \quad \textsf{$n_c \leftarrow len(\dot{y})$}\\
	3. & \quad \textsf{$b \leftarrow 1$}\\
	4. & \quad \textsf{$(\delta_1,b) \leftarrow$ \textsc{Get-Position}($\tau$, $b$, $\dot{y}$, $1$)}\\
	5. & \quad \textsf{for $i \leftarrow 2$ to $n_c$ do}\\
	6. & \quad \qquad \textsf{$(\delta_i,b) \leftarrow$ \textsc{Get-Position}($\tau$, $b$, $\dot{y}$, $i$)}\\
	7. & \quad \qquad \textsf{$\bar{y}[i-1] \leftarrow \delta(i)-\delta(i-1)$}\\
	8. & \quad \textsf{return $\bar{y}$}\\
\end{tabular}

\end{tabular}
\end{center}
\caption{\label{code:getpos} (On the left) The pseudocode of procedure \textsc{Get-Position} which computes the index of the block corresponding the $i$-th occurrence of the pivot character in $y$. (On the right) The pseudocode of procedure \textsc{Compute-Character-Distance-Sampling} which computes, on the flight from $\dot{y}$, the Characters Distance sampled version of $y$.}
\end{figure}

The following lemma introduces an efficient way for computing $\bar{y}$ from $\dot{y}$ and $\tau$.

\begin{lemma} \label{lemma1} Let $y$ be a text of length $n$, let $c\in \Sigma$ be the pivot character and assume $c$ occurs $n_c$ times in $y$. In addition let $\dot{y}$ and $\bar{y}$ be the $k$-bounded-position and character-distance sampled versions of $y$, respectively. If $b = \min\{ j : \tau[j]\geq i+1\}$ and $a = \min\{ j : \tau[j]\geq i\}$, then the following relation holds
\begin{equation}\label{eq:2}
\bar{y}[i] = \dot{y}[i+1]+ (\tau[a]-\tau[b])k - \dot{y}[i]
\end{equation}
\end{lemma}
\begin{proof}
Le $b = \min\{ j : \tau[j]\geq i+1\}$ and $a = \min\{ j : \tau[j]\geq i\}$. By corollary \ref{corol:position} we have
$$
\begin{array}{rll}
	\bar{y}[i] = & \Delta(i) = & \textrm{ by Definition \ref{def:dist}}\\
			& \delta(i+1) - \delta(i) = & \\
		         & (\tau[b]-1) k + \dot{y}[i+1] - (\tau[a]-1) k + \dot{y}[i] = & \textrm{ by Corollary \ref{corol:position}}\\
			& \dot{y}[i+1]+ (\tau[a]-\tau[b])k - \dot{y}[i] & \\
\end{array}
$$
proving the Lemma.\QEDA
\end{proof}

Based on the formula introduced by Lemma \ref{lemma1} the pseudocode shown in Figure \ref{code:getpos} (on the right) describes a procedure for computing on the flight from $\dot{y}$ the character distance sampled version of the text $y$. It is easy to prove that the time complexity of this procedure is given by $O(n_c + n/k)$, where a time $O(n_c)$ is required for scanning the sequence $\dot{y}$, while a time $O(n/k)$ is required for scanning the mapping table $\tau$. If we suppose equiprobability and independence of characters, the number of expected occurrences of the pivot character is $\mathbf{E}(n_c)= n/\sigma$. Thus under the assumption\footnote{In practical cases we can implement our solution with a block size $k=256$, which allows to represent the elements of the sequence $\dot{y}$ using a single byte. In such a case the assumption $k\geq \sigma$ is plausible for any practical application.} that $k\geq \sigma$ the overall average time complexity of this procedure is $O(n_c)$. Its worst-case time complexity is plainly $O(n)$.

\smallskip

Under the specific assumption that $\Delta(i)<k$, for each $1\leq i\leq n_c-1$, the sequence $\bar{y}$ can be computed in a more efficient way. It is defined by the following Corollary \ref{corol1} which trivially follows from Lemma \ref{lemma1}.

\begin{corollary} \label{corol1} Let $y$ be a text of length $n$, let $c\in \Sigma$ be the pivot character and assume $c$ occurs $n_c$ times in $y$. In addition let $\dot{y}$ and $\bar{y}$ be the $k$-bounded-position and character-distance sampled versions of $y$, respectively. Then, if we assume that $\Delta(i)<k$, for each $1\leq i <n_c$, the following relation holds
\begin{equation}\label{eq:2}
\bar{y}[i] = \left\{ \begin{array} {ll}
			\dot{y}[i+1] - \dot{y}[i] & \textrm{ if } \dot{y}[i+1] > \dot{y}[i]\\
			\dot{y}[i+1]+k - \dot{y}[i] & \textrm{ otherwise }\\
			\end{array}
		\right.
\end{equation}
\end{corollary}
\begin{proof}
	Assume that the $i$-th occurrence of the pivot character is in the $j$-th block. Since $\Delta(j)<k$ by definition, the $(i+1)$-th occurrence of the pivot character is in block $j$ or (at most) in block $(j+1)$. 

	Plainly, if $\dot{y}[i+1] \leq \dot{y}[i]$, the two pivot characters occurs into different (consecutive) blocks. Thus $\Delta(i) = \dot{y}[i+1]+k - \dot{y}[i]$.

	Conversely, assume now that $\dot{y}[i+1] > \dot{y}[i]$. The $(i+1)$-th occurrence of the pivot character cannot be in block $j+1$, because is such a case we should have $\Delta(i) = \dot{y}[i+1]+k-\dot{y}[i] > k$. Thus the $(i+1)$-th occurrence of the pivot character is in block $j$ and $\Delta(i) = \dot{y}[i+1]-\dot{y}[i] > k$.\QEDA

\end{proof}

We are now ready to describe the preprocessing and the searching phase of our new proposed algorithm.


\begin{figure}
\begin{center}
\begin{tabular}{ll}

\begin{tabular}{ll}
	\multicolumn{2}{l}{\textsc{Compute-Distance-Sampling}($y$, $n$, $\bar{\Sigma}$)}\\ 
	1. & \quad \textsf{$\bar{y} \leftarrow \langle \rangle$}\\
	2. & \quad \textsf{$j \leftarrow 0$}\\
	3. & \quad \textsf{$p \leftarrow 0$}\\
	4. & \quad \textsf{for $i \leftarrow 1$ to $n$ do}\\
	5. & \quad \quad \textsf{if $y[i] \in \bar{\Sigma}$ then}\\
	6. & \quad \quad \quad \textsf{$j \leftarrow j+1$}\\
	7. & \quad \quad \quad \textsf{$\bar{y}[j] \leftarrow i-p$}\\
	8. & \quad \quad \quad \textsf{$p \leftarrow i$}\\
	9. & \quad \textsf{return ($\bar{y}$, $j$)}\\
	&\\
	&\\
	&\\
\end{tabular}
&
\begin{tabular}{ll}
	\multicolumn{2}{l}{\textsc{Compute-Position-Sampling}($y$, $n$, $\bar{\Sigma}$, $k$)}\\ 
	1. & \quad \textsf{$\dot{y} \leftarrow \langle \rangle$}\\
	2. & \quad \textsf{$\tau \leftarrow$ a table of $\lceil n/k \rceil$ entries}\\
	3. & \quad \textsf{for $i \leftarrow 1$ to $\lceil n/k\rceil$ do $\tau[i] \leftarrow 0$}\\
	4. & \quad \textsf{$j \leftarrow 0$}\\
	5. & \quad \textsf{for $i \leftarrow 1$ to $n$ do}\\
	6. & \quad \quad \textsf{if $\tau[\lceil i/k \rceil] = 0$ then}\\
	7. & \quad \quad \quad \textsf{$\tau[\lceil i/k \rceil] \leftarrow \tau[\lceil i/k \rceil-1]$}\\
	8. & \quad \quad \textsf{if $y[i] \in \bar{\Sigma}$ then}\\
	9. & \quad \quad \quad \textsf{$j \leftarrow j+1$}\\
	10. & \quad \quad \quad \textsf{$\dot{y}[j] \leftarrow  i \mod k $}\\
	11. & \quad \quad \quad \textsf{$\tau[\lceil i/k \rceil] \leftarrow j$}\\
	12. & \quad \textsf{return ($\dot{y}$,$j$,$\tau$)}\\
\end{tabular}

\end{tabular}
\end{center}
\caption{\label{code:preprocessing} (On the left) The pseudocode of procedure \textsc{Compute-Distance-Sampling} for the construction of the \emph{character distance sampling} version of a text $y$. (On the right) The pseudocode of procedure \textsc{Compute-Position-Sampling} for the construction of the \emph{character position sampling} version of a text $y$}
\end{figure}


\newpage

\subsection{The Preprocessing Phase}

Assuming that the maximum distance between two consecutive occurrences of the pivot character is bounded by $k$, by Definition \ref{def:pos} and by Definition \ref{def:dist}, the sequences $\dot{y}$ and $\bar{y}$ both require $(n_c)\log(k)$ bits to be maintained. In practical cases we can chose $k=256$ using a single byte to store each value of the sampled sequences and maintaining the assumption $\Delta(i)<k$ more feasible.  This will allow us to store the sampled text using only $n_c$ bytes. 

Fig. \ref{exp_dist} reports the maximum and average distances between two consecutive occurrences, computed for the most frequent characters in a natural language text. Observe that the first 6 most frequent characters follow the constraint on the maximum distance.

As a consequence the choice of the  pivot character directly influences the additional memory used for storing the sampled text (the larger is the rank of the pivot character the shorter is the resulting sampled text) and the performances of the searching phase, as we will see later in Section \ref{sec:results}.

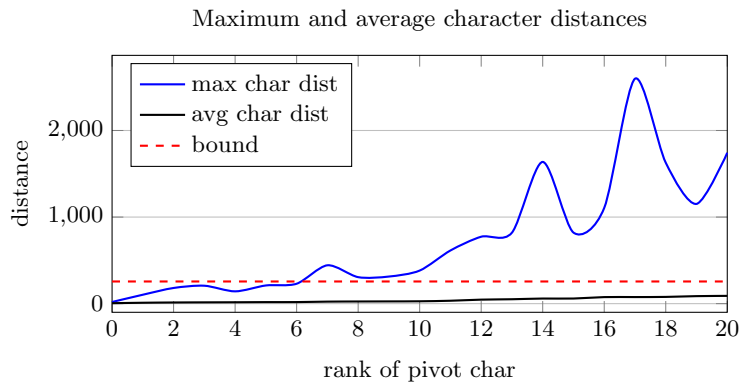
\begin{figure}
\begin{center}
\begin{tikzpicture}
	\begin{axis}[width={0.80\textwidth},xmin=0, xmax=20, ymin=-100, height={5cm}, ymajorgrids = true, title={Maximum and average character distances}, xlabel={rank of pivot char}, ylabel={distance}, legend cell align={left}, legend pos={north west}]
\addplot[thick,blue,smooth]  coordinates {
(0, 19) (1, 103) (2, 181) (3, 207) (4, 142) (5, 210) (6, 230) (7, 443) (8, 306) (9, 313) (10, 382) (11, 614) (12, 773) (13, 820) (14, 1637) (15, 821) (16, 1106) (17, 2598) (18, 1628) (19, 1152) (20, 1743) 
};
\addplot[thick,black,smooth]  coordinates {
(0, 5.25) (1, 10.34) (2, 13.41) (3, 14.90) (4, 15.93) (5, 17.65) (6, 17.98) (7, 23.31) (8, 24.92) (9, 25.81) (10, 27.20) (11, 33.46) (12, 46.21) (13, 50.83) (14, 58.87) (15, 59.04) (16, 75.50) (17, 75.87) (18, 77.79) (19, 86.74) (20, 89.75) 
};
\addplot[thick,red, domain=0:20,dashed] ({\x},{256});
\legend {max char dist, avg char dist, bound}
\end{axis}
\end{tikzpicture}\\
\end{center}
\caption{\label{exp_dist}Maximum and average distances between two consecutive occurrences, computed for the most frequent characters in a natural language text. On the $x$ axis characters are ordered on the base of their rank value, in a non decreasing order. The red line represents the bound $k=256$.}
\end{figure}

Let $y$ be an input text of length $n$ over the alphabet $\Sigma$, of size $\sigma$, and let $r$ be a constant input parameter. Specifically $r$ is the rank of the pivot character to be selected for building the sampled text. The rank $r$ must be chosen in order to have $\Delta(i)<k$, for each $1\leq i\leq n_c$.

The first step of the preprocessing phase of the algorithm consists in counting the frequencies of each character $c\in\Sigma$ and in computing their corresponding ranks. 

Subsequently the algorithm builds and store the $k$-bounded position sampled text $\dot{y}$. This step requires $O(r\log(n) + n)$-time (since only the first $r$ characters need to be ordered) and $O(n_c)$-space, where $n_c$ is the number of occurrences of the pivot character in $y$.

Each element of the mapping table $\tau$ can be stored using $\log(n)$ bit. Therefore the overall space complexity of the algorithm is $O(n_c+n\log(n)/k)$.


\begin{figure}
\begin{center}

\begin{tabular}{ll}
	\multicolumn{2}{l}{\textsc{Verify}($x$, $m$, $y$, $s$)}\\ 
	1. & \quad \textsf{$i \leftarrow 1$}\\
	2. & \quad \textsf{while $i \leq m$ and $x[i]=y[s+i]$ do $i \leftarrow i+1$}\\
	3. & \quad \textsf{if $i > m$ then return \textsc{True}}\\
	4. & \quad \textsf{return \textsc{False}}\\
\end{tabular}

\end{center}
\caption{\label{code:verification} The pseudocode of procedure \textsc{Verify} for testing the occurrence of a pattern $x$, of length $m$, in a text $y$, of length $n$, starting at position $s$.}
\end{figure}

\subsection{The Searching Phase}
Let $x$ be an input pattern of length $m$ and let $c \in \Sigma$ be the pivot character. Let $m_c$ be the number of occurrences of the pivot character in $x$.
The searching phase can be then divided in three different subroutines, depending on the value of $m_c$.
All searching procedures work using a filtering approach. They takes advantage of the partial-index precomputed during the preprocessing phase in order to quickly locate any candidate substring of the text which may include an occurrence of the pattern.

If such candidate substring has length $m$ the algorithm simply performs a character-by-character comparison between the pattern and the substring.  Figure \ref{code:verification} shows the pseudocode of procedure \textsc{Verify} used by the algorithm for testing the occurrence of a pattern $x$, of length $m$, in a text $y$, of length $n$, starting at position $s$. Its complexity is plainly $O(m)$ in the worst case.

Otherwise if the candidate substring has length greater than $m$, then a searching procedure is called, based on any exact online string matching algorithm. As we will discuss later, if we suppose that the underlying algorithm has a linear worst case time complexity, as in the case of the KMP algorithm \cite{Ref1}, then our solution achieves the same complexity in the worst case. Similarly if we suppose to implement the searching procedure using an optimal average string matching algorithm, like the BDM algorithm \cite{Ref4}, the resulting solution achieves an optimal $O(n\log_{\sigma}m/m)$ average time complexity.

In what follows we describe in details the three different searching procedures which are applied when $m_c=0$, $m_c=1$ and $m_c>1$, respectively.

\begin{figure}
\begin{center}
\begin{tabular}{ll}
	\multicolumn{2}{l}{\textsc{Search-0}($x$, $\dot{y}$, $y$)}\\ 
	$\Delta$ & \quad \textsf{We assume $c$ does not occur in $x$}\\	
	1. & \quad \textsf{$m \leftarrow len(x)$}\\
	2. & \quad \textsf{$n_c \leftarrow len(\dot{y})$}\\
	3. & \quad \textsf{$b \leftarrow 1$}\\
	4. & \quad \textsf{$\delta_0 \leftarrow 0$}\\
	5. & \quad \textsf{for $i \leq 2$ to $n_c$ do}\\
	6. & \quad \qquad \textsf{$(\delta_i,b) \leftarrow \textsc{Get-Position}(\tau, b, \dot{y}, i)$}\\
	7. & \quad \qquad \textsf{if ($\delta_i - \delta_{i-1}> m$) then}\\
	8. & \quad \qquad \qquad \textsf{search for $x$ in $y[\delta_{i-1}+1 .. \delta_{i}-1]$}\\
	9. & \quad \textsf{if ($n+1 - \delta_{n_c}> m$) then}\\
	10. & \quad \qquad \textsf{search for $x$ in $y[\delta_{n_c}+1 .. n]$}\\
\end{tabular}
\end{center}
\caption{\label{code:search0} The pseudocode of procedure \textsc{Search-0} for the sampled string matching problem, when the pivot character does not occur in the input pattern $x$.}
\end{figure}

\subsubsection{Case 1: $m_c=0$}
Consider firstly  the case where the pattern does not contain any occurrence of the pivot character $c$, so that  $m_c$  is equal to zero. Under this assumption the algorithm searches for the pattern $x$ in all substrings of the original text which do not contain the pivot character. Specifically such substrings are identified in the original text by the intervals $[\delta(i)+1 .. \delta(i+1)-1]$, for each $0\leq i \leq n_c$, assuming $\delta(0)=0$ and $\delta(n_c+1)=n+1$.

Figure \ref{code:search0} shows the pseudocode of the algorithm which searches for all occurrences of a pattern $x$, when the pivot character $c$ does not occur in it.  Specifically, for each $1 \leq i \leq n_c+1$, the algorithm checks if the interval $\delta(i)-\delta(i-1)$ is greater than $m$ (lines 5-7). In such a case the algorithm searches for $x$ in the substring of the text $y[\delta(i-1)+1..\delta(i)-1]$ (line 8) using any standard string matching algorithm. Otherwise the substring is skipped, since no occurrence of the pattern could be found at such position.

The following theorem proves that procedure \textsc{Search-0} achieves optimal worst-case and average-case time complexity.

\begin{theorem}\label{th:search0}
Le $x$ and $y$ be two strings of size $m$ and $n$, respectively, over an alphabet $\Sigma$ of size $\sigma>1$. Let $c \in \Sigma$ be the pivot character and let $\dot{y}$ be the $k$-bounded position sampled version of the text $y$. Under the assumption of equiprobability and independence of characters in $\Sigma$, the worst-case and average time complexity of the \textsc{Search-0} algorithm are $O(n)$ and $O(n\log_{\sigma}m/m)$, respectively.
\end{theorem}
\begin{proof}
In order to evaluate the worst-case time complexity of procedure \textsc{Search-0}, we can notice that each substring of the text is scanned at least once in line 10, with no overlap. Thus if we use a linear algorithm to perform the standard search then it is trivial to prove that the whole searching procedure requires
$$
	T^0_{\textsf{wst}}(n) = O(m) + \sum_{i=1}^{n_c-1} O\left(\Delta(i)\right) + O\left(\frac{n}{k}\right) = O(n).
$$
Assuming that the underlying algorithm has an $O(n\log m/m))$ average time complexity, on a text of length $n$ and a pattern of length $m$, we can express the expected average time complexity as
$$
	T^0_{\textsf{avg}}(n) = \sum_{i=1}^{n_c-1} O\left(\frac{\Delta(i)\log_{\sigma} m}{m}\right) + O\left(\frac{n}{k}\right) = O\left(\frac{n\log_{\sigma} m}{m}\right).
$$
for enough great values $k\geq m/\log_{\sigma} m$ and $k\geq \sigma$.\QEDA
\end{proof}
Observe that, in the case of a natural language text, where generally $\sigma\geq 100$, a block size $k=256$ is enough to reach the optimal average time complexity for any (short enough) pattern of length $m<256$.

\subsubsection{Case 2: $m_c=1$}
Suppose now the case where the pattern $x$ contains a single occurrence of the pivot character $c$, so that the length of the sampled version of the pattern is still equal to zero. The algorithm efficiently takes advantage of the information precomputed in $\dot{y}$ using the positions of the pivot character $c$ in $y$ as an anchor to locate all candidate occurrences of $x$.

Figure \ref{code:search1} shows the pseudocode of the algorithm which searches for all occurrences of a pattern $x$, when the pivot character $c$ occurs only once in it.
Specifically, let $\alpha$ be the unique position in $x$ which contains the pivot character (line 3), i.e. we assume that $x[\alpha]=c$ and $x[1..\alpha-1]$ does not contain $c$. Then, for each $0\leq i \leq n_c-1$, the algorithm checks if the interval $\delta(i-1)-\delta(i-2)$ is greater than $\alpha-1$ and if the interval $\delta(i)-\delta(i-1)$ is greater than $m-\alpha$ (lines 7-9). In such a case the algorithm merely checks if the substring of the text $y[\delta(i)-\alpha+1..\delta(i)+m-\alpha]$ is equal to the pattern (line 10). Otherwise the substring is skipped. As before we assume that $\delta(0) = 0$ (line 5) and $\delta(n_c+1)= n+1$ (line 11). 
The last alignment of the pattern in the text is verified separately at the end of the main cycle (lines 11-12). 

\begin{figure}
\begin{center}
\begin{tabular}{ll}
	\multicolumn{2}{l}{\textsc{Search-1}($x$, $\dot{y}$, $y$)}\\ 
	$\Delta$ & \quad \textsf{We assume $c$ occurs once in $x$}\\	
	1. & \quad \textsf{$m \leftarrow len(x)$}\\
	2. & \quad \textsf{$n_c \leftarrow len(\dot{y})$}\\
	3. & \quad \textsf{$\alpha \leftarrow \min\{i:x[i]=c\}$}\\
	4. & \quad \textsf{$b \leftarrow 1$}\\
	5. & \quad \textsf{$\delta_0\leftarrow 0$}\\
	6. & \quad \textsf{$(\delta_1,b) \leftarrow \textsc{Get-Position}(\tau, b, \dot{y}, 1)$}\\
	7. & \quad \textsf{for $i \leq 2$ to $n_c$ do}\\
	8. & \quad \qquad \textsf{$(\delta_i,b) \leftarrow \textsc{Get-Position}(\tau, b, \dot{y}, i)$}\\
	9. & \quad \qquad \textsf{if ($\delta_{i-1} - \delta_{i-2}> \alpha-1$ and $\delta_{i} - \delta_{i-1}> m-\alpha$) then}\\
	10. & \quad \qquad \qquad \textsf{\textsc{Verify}($x$, $m$, $y$, $p_{i-1}-\alpha+1$)}\\
	11. & \quad \textsf{if ($\delta_{n_c} - \delta_{n_c-1}> \alpha-1$ and $n+1 - \delta_{n_c}> m-\alpha$) then}\\
	12. & \quad \qquad \textsf{\textsc{Verify}($x$, $m$, $y$, $\delta_{n_c}-\alpha+1$)}\\
\end{tabular}
\end{center}
\caption{\label{code:search1} The pseudocode of procedure \textsc{Search-1} for the sampled string matching problem, when the pivot character occurs ones in the input pattern $x$.}
\end{figure}

The following theorem proves that procedure \textsc{Search-1} achieves optimal worst-case and average-case time complexity.

\begin{theorem}\label{th:search1}
Le $x$ and $y$ be two strings of size $m$ and $n$, respectively, over an alphabet $\Sigma$ of size $\sigma>1$. Let $c \in \Sigma$ be the pivot character and let $\dot{y}$ be the $k$-bounded position sampled version of the text $y$. Under the assumption of equiprobability and independence of characters in $\Sigma$, the worst-case and average time complexity of the \textsc{Search-1} algorithm are $O(n)$ and $O(n\log_{\sigma}m/m)$, respectively.
\end{theorem}
\begin{proof}
In order to evaluate the worst-case time complexity of the algorithm notice that each character could be involved in, at most, two consecutive checks in line 10.  Specifically any text position in the interval $[\delta(i-1)+1..\delta(i)-1]$ could be involved in the verification of the substrings $y[\delta(i-1)-\alpha+1 .. \delta(i-1)+m-\alpha]$ and $y[\delta(i)-\alpha+1 .. \delta(i)+m-\alpha]$. 
Thus the overall worst case time complexity of the searching phase is 
$T^1_{\textsf{wst}}(n) = O(n)$.

In order to evaluate the average-case time complexity of the algorithm notice that the expected number of occurrences in $y$ of the pivot character is given by $\mathbf{E}(n_c) = n/\sigma$. Moreover, for any candidate occurrence of $x$ in $y$, the number $\mathbf{E}(\textrm{insp})$ of expected character inspections performed by procedure \textsc{Verify}, when called on a pattern of length $m$, is given by 
$$
\mathbf{E}(\textrm{insp}) = 1+ \sum_{i=1}^{m-1}\left(\frac{1}{\sigma}\right)^i \leq \frac{\sigma}{\sigma-1}
$$
Thus the average time complexity of the algorithm ican be expressed by
$$
\begin{array}{ll}
	T^1_{\textsf{avg}}(n) = & \mathbf{E}(n_c) \cdot \mathbf{E}(\textrm{insp}) =\\
				   & O\left(\frac{n}{\sigma}\right) \cdot O\left(\frac{\sigma}{\sigma-1}\right) = \\
				   & O\left( \frac{n}{\sigma-1} \right)\\
\end{array}
$$
obtaining the optimal average time complexity $O(n\log_{\sigma}m/m))$ for great enough alphabets of size $\sigma>(m/\log_{\sigma} m)+1$, and for $k\geq \sigma$.\QEDA
\end{proof}

\subsubsection{Case 3: $m_c\geq 2$}
If the number of occurrences of the pivot character in $x$ are  more than $2$ (i.e. if $m_c\geq 2$ and $\bar{m}\geq 1$) then the algorithm uses the sampled text $\bar{y}$ to compute on the fly the character-distance sampled version of $y$ and use it to searching for any occurrence of $\bar{x}$. This is used as a filtering phase for locating in $y$ any candidate occurrence of $x$.

Figure \ref{code:search2} shows the pseudocode of the algorithm which searches for all occurrences of a pattern $x$, when the pivot character $c$ occurs more than once in it.
First the character distance sampled version of the pattern $\bar{x}$ is computed (line 5). Then the algorithm searches for $\bar{x}$ in the $\bar{y}$ using any exact online string matching algorithm (line 6). Notice that $\bar{y}$ can be efficiently retrieved online from the sampled text $\dot{y}$, using relation given in (\ref{eq:2}).

For each occurrence position $j$ of $\bar{x}$ in $\bar{y}$ an additional procedure must be run  to check if such occurrence corresponds to a match of the whole pattern $x$ in $y$ (lines 7-9). For this purpose the algorithm checks if the substring of the text $y[\delta(j)-\alpha ... \delta(j)-\alpha+m-1]$ is equal to $x$, where $\alpha$, as before, is the first position in $x$ where the pivot character occurs.
Notice that the value of $\delta(j)$ can be obtained in constant time from $\dot{y}[j]$ 
(line 8).

\begin{figure}
\begin{center}
\begin{tabular}{ll}
	\multicolumn{2}{l}{\textsc{Search-2$^+$}($x$, $\dot{y}$, $y$)}\\ 
	$\Delta$ & \quad \textsf{We assume $c$ occurs more than once in $x$}\\	
	1. & \quad \textsf{$m \leftarrow len(x)$}\\
	2. & \quad \textsf{$\dot{n} \leftarrow len(\dot{y})$}\\
	3. & \quad \textsf{$\alpha \leftarrow \min\{i:x[i]=c\}$}\\
	4. & \quad \textsf{$b \leftarrow 1$}\\
	5. & \quad \textsf{$(\bar{x}, \bar{m}) \leftarrow \textsc{Compute-Distance-Sampling}(x,m,\{c\})$}\\
	6. & \quad \textsf{search for $\bar{x}$ in $\bar{y}$ :}\\
	7. & \quad \qquad \textsf{for each $i$ such that $\bar{x} = \bar{y}[i .. i+\bar{m}-1]$ do}\\
	8. & \quad \qquad \qquad \textsf{$(\delta_i,b) \leftarrow \textsc{Get-Position}(\tau, b, i)$}\\
	9. & \quad \qquad \qquad \textsf{\textsc{Verify}($x$, $m$, $y$, $\delta_{i}-\alpha$)}\\
\end{tabular}
\end{center}
\caption{\label{code:search2} The pseudocode of procedure \textsc{Search-2} for the sampled string matching problem, when the pivot character occurs more than ones in the input pattern $x$.}
\end{figure}

The following theorem proves that procedure \textsc{Search-2} achieves optimal worst-case and average-case time complexity.

\begin{theorem}\label{th:search2}
Le $x$ and $y$ be two strings of size $m$ and $n$, respectively, over an alphabet $\Sigma$ of size $\sigma>1$. Let $c \in \Sigma$ be the pivot character and let $\dot{y}$ be the $k$-bounded position sampled version of the text $y$. Under the assumption of equiprobability and independence of characters in $\Sigma$, the worst-case and average time complexity of the \textsc{Search-2} algorithm are $O(n)$ and $O(n\log_{\sigma}m/m)$, respectively.
\end{theorem}
\begin{proof}
In order to evaluate the worst-case time complexity of the algorithm in this last case notice that, if we use a linear algorithm to search $\bar{y}$ for $\bar{x}$, the overall time complexity of the searching phase is $O(n_c + n_x m + n/k)$, where $n_x$ is the number of occurrences of $\bar{x}$ in $\bar{y}$. In the worst case it translates in $O(n_c m)$ worst case time complexity. However it is not  difficult to suppose to implement procedure \textsc{Verify} based on a linear algorithm, as KMP, in order to remember all positions of the text which have been already verified, allowing the algorithm to run in overall $T^{^{2+}}_{\textsf{wst}}(n)=O(n)$ worst-case time complexity.

In order to evaluate the average-case time complexity of the algorithm notice that time required for searching $\bar{x}$ in $\bar{y}$ is $O(n_c\log m_c/m_c) + n/k)$. Moreover, observe that the number of verification is bounded by the expected number of occurrences of the pivot character in $y$, thus, following the same line of Theorem \ref{th:search1}, the overall average time complexity of the verifications phase is $O(n/(\sigma-1))$.
Thus the average time complexity of the algorithm can be expressed by
$$
	T^{^{2+}}_{\textsf{avg}}(n) = O\left(\frac{n_c\log m_c}{m_c}\right) + O\left( \frac{n}{k}\right) + O\left( \frac{n}{\sigma-1} \right)
$$
obtaining the optimal average time complexity $O(n\log_{\sigma}m/m)$ for great enough values of $k$ and $\sigma$, such that $k\geq \sigma \geq (m/\log_{\sigma} m)+1$.\QEDA
\end{proof}


\section{Experimental Results} \label{sec:results}
In this section we report the results of an extensive evaluation of the new presented algorithm in comparison with the \textsc{Ots} algorithm by Claude \emph{et al.} \cite{Ref37} for the sampled string matching problem.
The algorithms have been implemented in \textsf{C}, and have been tested using a variant of the \textsc{Smart} tool~\cite{FTBDM16}, properly tuned for testing string matching algorithms based on a text-sampling approach, and executed locally on a MacBook Pro with 4 Cores, a 2 GHz Intel Core i7 processor, 16 GB RAM 1600 MHz DDR3, 256 KB of L2 Cache and 6 MB of Cache L3.\footnote{The \textsc{Smart} tool is available online for download at  \url{http://www.dmi.unict.it/~faro/smart/} or at  \url{https://github.com/smart-tool/smart}.}
Although all the evaluated algorithms could be implemented in conjunction with any online string matching algorithm, for the underlying searching procedure, for a fair comparison in our tests we realise implementations in conjunction with the Horspool algorithm, following the same line of the experimental evaluation conducted in \cite{Ref37}.

Comparisons have been performed in terms of preprocessing times, searching times and space consumption.
For our tests, we used the English text of size 5MB provided by the research tool \textsc{Smart}, available online for download.\footnote{Specifically, the text buffer is the concatenation of two different texts: The King James version of the bible (3.9 MB) and The CIA world fact book (2.4 MB). The first 5MB of the resulting text buffer have been used in our experimental results.}  

In the experimental evaluation, patterns of length $m$ were randomly extracted from the text (thus the number of reported occurrences is always greater than $0$), with $m$ ranging over the set of values $\{2^i \mid 2\leq i \leq 8\}$.
In all cases, the mean over the running times (expressed in milliseconds) of $1000$ runs has been reported.

\subsection{Space requirements}
In the context of text-sampling string-matching space requirement is one of the most significant parameter to take into account. It indicates how much additional space, with regard to the size of the text, is required by a given solution to solve the problem.

Conversely to other kind of algorithms, like compressed-matching algorithms, which are allowed to maintain the input strings in a compressed form (with a detriment in processing time), algorithms in text-sampling matching require to store the whole text together with the additional sampled-text which is used to speed-up the searching phase. In this context they are much more similar to online string-matching solutions and, although they have the additional good property to allow a direct access to the input text, to be of any practical interest they should require as little extra space as possible.

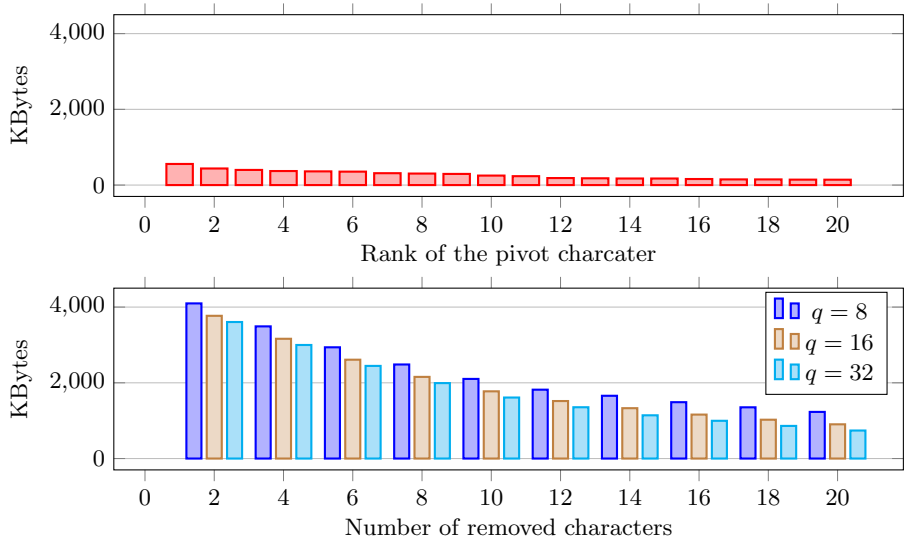
\begin{figure}[!t]
\begin{center}
\begin{tikzpicture}
	\begin{axis}[
		ybar,
		width={0.99\textwidth},
		height={4cm},
		ymax=4500,
		ymin=-300,
		xlabel=Rank of the pivot charcater,
		ylabel=KBytes,
		ymajorgrids = true,
		]
	\addplot[thick, color=red,fill=red!30] coordinates {
(1, 557.45)
(2, 438.11)
(3, 399.89)
(4, 370.58)
(5, 361.65)
(6, 354.82)
(7, 313.59)
(8, 304.11)
(9, 294.52)
(10, 250.63)
(11, 234.97)
(12, 185.23)
(13, 179.27)
(14, 173.38)
(15, 173.23)
(16, 159.83)
(17, 149.67)
(18, 149.42)
(19, 143.44)
(20, 141.98)
};
	\end{axis}
\end{tikzpicture}\\
\begin{tikzpicture}
	\begin{axis}[
		ybar,
		width={0.99\textwidth},
		height={4cm},
		xlabel=Number of removed characters,
		xmin=1,
		ylabel=KBytes,
		ymax=4500,
		ymin=-300,
		ymajorgrids = true,
		bar width=0.2cm,
		enlarge x limits=true,
		]
	\addplot[thick, color=blue,fill=blue!30] coordinates {
(	2  	,	4097.08	)
(	4  	,	3490.46	)
(	6  	,	2937.84	)
(	8  	,	2483.99	)
(	10  	,	2102.68	)
(	12  	,	1818.64	)
(	14  	,	1657.52	)
(	16  	,	1488.31	)
(	18  	,	1353.07	)
(	20  	,	1231.50	)
};
	\addplot[thick, color=brown,fill=brown!30] coordinates {
(	2  	,	3769.40	)
(	4  	,	3162.78	)
(	6  	,	2610.16	)
(	8  	,	2156.30	)
(	10  	,	1775.00	)
(	12  	,	1518.64	)
(	14  	,	1329.84	)
(	16  	,	1160.63	)
(	18  	,	1025.39	)
(	20  	,	903.82	)
};
	\addplot[thick, color=cyan,fill=cyan!30] coordinates {
(	2  	,	3605.56	)
(	4  	,	2998.94	)
(	6  	,	2446.32	)
(	8  	,	1992.46	)
(	10  	,	1611.16	)
(	12  	,	1354.81	)
(	14  	,	1141.98	)
(	16  	,	996.79	)
(	18  	,	861.55	)
(	20  	,	739.98	)
};
\legend{$q=8$,$q=16$,$q=32$}
	\end{axis}
\end{tikzpicture}\\
\end{center}
\caption{\label{fig:space}Space consumption of the text sampled algorithms. All values are in KB. On the top: memory space required by the new algorithm, for different pivot characters with rank ranging from $1$ to $20$.
On the bottom: memory space required by the \textsc{Ots} algorithm, for different sets of removed characters, with size ranging from $2$ to $20$, and for different values of the parameter $q$.}
\end{figure}

Table \ref{fig:space} shows the space consumption of the text-sampling algorithms described in this paper. Specifically, memory space required by the new algorithm is plotted on variations of the pivot character, while memory space required by the \textsc{Ots} algorithm is plotted on variations of the size of the set of removed characters. Both values range from $2$ to $20$.

Space consumption of the \textsc{Ots} algorithm ranges from $4$ MB (80\% of text size), when only 2 characters are removed and $q=8$, to  $740$ KB ($14.8$\% of text size), in the case of 20 removed characters and $q=32$. As expected, the function which describes memory requirements follows a decreasing trend while the value of $\sigma^*$ increases. 

A similar trend can be observed for our new algorithm, whose space consumption decreases while the rank of the pivot character increases. However, in this case, space requirement ranges from $557$ KB ($11$\% of text size), if we select the most frequent character as a pivot, to $142$ KB ($2.8$\% of text size), when the pivot is the character with rank position $20$. 
Thus the benefit in space consumption obtained by the new algorithm ranges (at least) from  $24$\% to $80$\%, even with the best performances of the  \textsc{Ots} algorithm.

Observe also that in the case of short texts (with a size of few Mega Bytes), the space consumption of our algorithm is easily comparable with space requirements of a standard online string matching algorithm.  
For instance, if we suppose to search for a pattern of $30$ characters on a text of $3$ MB, the Boyer-Moore-Horspool algorithm \cite{Ref3} requires only $256$ Bytes (less than $0.1$\% of the text size), the BDM algorithm \cite{Ref4} requires approximately $8$ KB ($0.2$\% of text size), while more efficient solutions as the Berry-Ravidran and the WFR algorithms \cite{Ref35} require approximately $65$ KB ($2.1$\% of text size).
However, the space requirements of an online string matching algorithm does not depend on the size of the text, thus when the value of $n$ increases the gap may become considerable.

\subsection{Preprocessing and searching times}
In this section we compare the two approach for the text-sampling string matching problem, in terms of preprocessing and searching times. With the term \emph{preprocessing time} we refer to the time needed to construct the sampled text which will be used during the searching phase. We do not take into account in this measurements the preprocessing time needed by the online exact string matching algorithm used during the searching phase.
We will refer to \emph{searching time} as the time needed to perform searching of the pattern on both sampled and original texts, including any preprocessing of the underlying searching algorithm.

\begin{figure}
\begin{center}
\begin{tikzpicture}
	\begin{axis}[width={0.80\textwidth},ymin=0, height={5cm}, ymajorgrids = true,xtick={2,4,6,8,10,12,14,16,18,20}, title={Preprocessing times}, xlabel={rank of pivot char / number of removed chars}, ylabel={thousands of seconds}]
\addplot[thick,blue,smooth]  coordinates {
(1,22.51) (2,21.28) (3,20.68) (4,20.29) (5,20.21) (6,20.12) (7,19.63) (8,19.44) (9,19.38) (10,18.67) (11,18.44) (12,17.57) (13,17.61) (14,17.43) (15,17.45) (16,17.37) (17,17.08) (18,17.01) (19,17.06) (20,17.12)};
\addplot[thick,black,smooth]  coordinates {
(2,46.69) (4,47.70) (6,44.36) (8,39.55) (10,34.59) (12,29.90) (14,28.01) (16,25.43) (18,23.17) (20,21.66)};
\addplot[thick,gray,smooth]  coordinates {
(2,43.50) (4,44.51) (6,41.92) (8,37.58) (10,32.78) (12,29.09) (14,26.80) (16,24.52) (18,22.42) (20,20.99)};
\addplot[thick,brown,smooth]  coordinates {
(2,41.53) (4,42.81) (6,40.47) (8,36.36) (10,31.55) (12,28.36) (14,25.56) (16,23.87) (18,21.90) (20,20.44)};
\legend{new,$q=8$,$q=16$,$q=32$}
\end{axis}
\end{tikzpicture}\\
\end{center}
\caption{\label{exp_pre}Preprocessing times of the text sampled algorithms. Running times are expressed in thousands of seconds. The $x$ axis represents the rank of the pivot character in the case of the new algorithm, while represents the number of removed characters in the case of previous algorithms.}
\end{figure}
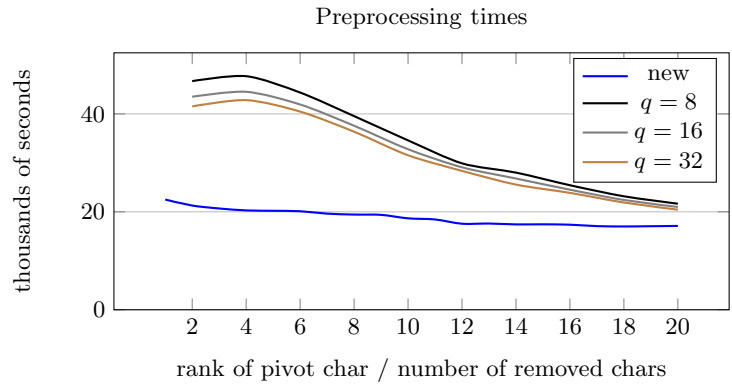

Fugure \ref{exp_pre} shows the preprocessing times of our algorithm together with those of the \textsc{Ots} algorithms (in the latter case we show preprocessing time for the three different values of $q\in \{8, 16, 32\}$).
It turns out that the new algorithm has always a faster preprocessing time, with a speed-up which goes from 50\%,  to 15\%, depending on the pivot element and on the number of removed characters. This is mainly correlated with size of the data structure constructed during the preprocessing phase and discussed in the previous section.


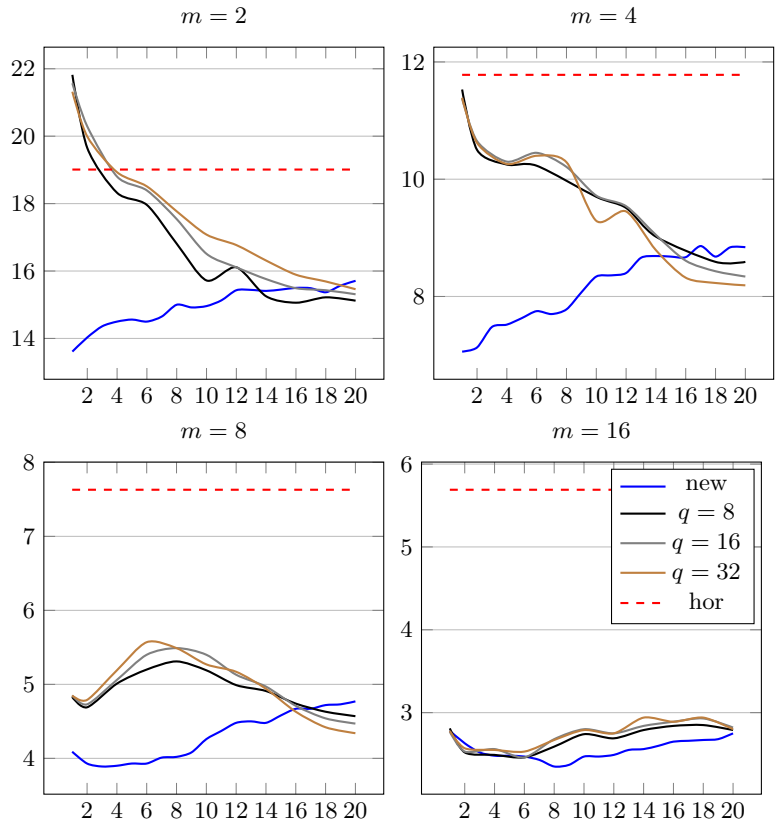
\begin{figure}[!t]
\begin{center}
\begin{tikzpicture}
	\begin{axis}[width={0.5\textwidth},height={6cm}, ymajorgrids = true,title={$m=2$},xtick={2,4,6,8,10,12,14,16,18,20}]
\addplot[thick,blue,smooth]  coordinates {
(1, 13.61) (2, 14.03) (3, 14.36) (4, 14.50) (5, 14.56) (6, 14.50) (7, 14.65) (8, 15.00) (9, 14.92) (10, 14.96) (11, 15.13) (12, 15.43) (13, 15.44) (14, 15.41) (15, 15.45) (16, 15.50) (17, 15.49) (18, 15.37) (19, 15.57) (20, 15.71)};
\addplot[thick,black,smooth]  coordinates {
(1,21.82) (2, 19.65) (4, 18.32) (6, 17.96) (8, 16.81) (10, 15.72) (12, 16.11) (14, 15.25) (16, 15.06) (18, 15.22) (20, 15.12)};
\addplot[thick,gray,smooth]  coordinates {
(1,21.57) (2, 20.29) (4, 18.80) (6, 18.39) (8, 17.54) (10, 16.51) (12, 16.11) (14, 15.76) (16, 15.49) (18, 15.43) (20, 15.31)};
\addplot[thick,brown,smooth]  coordinates {
(1,21.32) (2, 19.99) (4, 18.93) (6, 18.51) (8, 17.77) (10, 17.08) (12, 16.77) (14, 16.31) (16, 15.89) (18, 15.69) (20, 15.46)};
\addplot[thick,red, domain=1:20,dashed] ({\x},{19.01});
\end{axis}
\end{tikzpicture}
\begin{tikzpicture}
	\begin{axis}[width={0.5\textwidth},height={6cm}, ymajorgrids = true,title={$m=4$},xtick={2,4,6,8,10,12,14,16,18,20}]
\addplot[thick,blue,smooth]  coordinates {
(1,7.06)	(2,7.13)	(3,7.48)	(4,7.52)	(5,7.63)	(6,7.75)	(7,7.7)	(8,7.78)	(9,8.07)	(10,8.34)	(11,8.36)	(12,8.4)	(13,8.66)	(14,8.69)	(15,8.68)	(16,8.67)	(17,8.86)	(18,8.68)	(19,8.84)	(20,8.84)};
\addplot[thick,black,smooth]  coordinates {
(1,11.53)(2,10.5)	(4,10.25)	(6,10.23)	(10.12)	(10,9.7)	(12,9.51)	(14,9.02)	(8.76)	(18,8.59)	(20,8.59)};
\addplot[thick,gray,smooth]  coordinates {
(1,11.38)(2,10.65)	(4,10.3)	(6,10.45)	(8,10.21)	(10,9.72)	(12,9.54)	(14,9.05)	(16,8.61)	(18,8.43)	(20,8.34)};
\addplot[thick,brown,smooth]  coordinates {
(1,11.39)(2, 10.62)	(4, 10.26)	(6, 10.40)	( 8,10.29)	(10, 9.29)	(12, 9.45)	(14, 8.79)	( 16,8.32)	(18, 8.23)	(20, 8.19)};
\addplot[thick,red, domain=1:20,dashed] ({\x},{11.78});
\end{axis}
\end{tikzpicture}\\
\begin{tikzpicture}
	\begin{axis}[width={0.5\textwidth},height={6cm},	ymajorgrids = true,title={$m=8$},xtick={2,4,6,8,10,12,14,16,18,20}]
\addplot[thick,blue,smooth]  coordinates {
(1,4.09)	(2,3.93)	(3,3.89)	(4,3.90)	(5,3.93)	(6,3.93)	(7,4.01)	(8,4.02)	(9,4.08)	(10,4.26)	(11,4.37)	(12,4.48)	(13,4.5)	(14,4.48)	(15,4.58)	(16,4.67)	(17,4.67)	(18,4.72)	(19,4.73)	(20,4.77)};
\addplot[thick,black,smooth]  coordinates {
(1,4.83)(2,4.69)	(4,5.01)	(6,5.2)	(8,5.31)	(10,5.19)	(12,4.99)	(14,4.91)	(16,4.74)	(18,4.63)	(20,4.57)};
\addplot[thick,gray,smooth]  coordinates {
(1,4.85)(2,4.73)	(4,5.06)	(6,5.4)	(8,5.49)	(10,5.4)	(12,5.13)	(14,4.97)	(16,4.71)	(18,4.54)	(20,4.47)};
\addplot[thick,brown,smooth]  coordinates {
(1,4.85)(2,4.79)	(4,5.19)	(6,5.57)	(8,5.49)	(10,5.27)	(12,5.17)	(14,4.94)	(16,4.63)	(18,4.42)	(20,4.34)};
\addplot[thick,red, domain=1:20,dashed] ({\x},{7.63});
\end{axis}
\end{tikzpicture}
\begin{tikzpicture}
	\begin{axis}[width={0.5\textwidth},height={6cm},	ymajorgrids = true,title={$m=16$},xtick={2,4,6,8,10,12,14,16,18,20}]
\addplot[thick,blue,smooth]  coordinates {
(1,2.78)	(2,2.63)	(3,2.52)	(4,2.48)	(5,2.48)	(6,2.47)	(7,2.43)	(8,2.35)	(9,2.37)	(10,2.47)	(11,2.47)	(12,2.49)	(13,2.55)	(14,2.56)	(15,2.6)	(16,2.65)	(17,2.66)	(18,2.67)	(19,2.68)	(20,2.75)};
\addplot[thick,black,smooth]  coordinates {
(1,2.81) (2,2.52)	(4,2.49)	(6,2.46)	(8,2.59)	(10,2.74)	(12,2.69)	(14,2.79)	(16,2.84)	(18,2.85)	(20,2.79)};
\addplot[thick,gray,smooth]  coordinates {
(1,2.77) (2,2.53)	(4,2.56)	(6,2.46)	(8,2.68)	(10,2.8)	(12,2.75)	(14,2.84)	(16,2.89)	(18,2.93)	(20,2.82)};
\addplot[thick,brown,smooth]  coordinates {
(1,2.78) (2,2.57)	(4,2.55)	(6,2.53)	(8,2.67)	(10,2.79)	(12,2.75)	(14,2.94)	(16,2.89)	(18,2.94)	(20,2.8)};
\addplot[thick,red, domain=1:20,dashed] ({\x},{5.69});
\legend{new,$q=8$,$q=16$,$q=32$,hor}
\end{axis}
\end{tikzpicture}\\
\end{center}
\caption{\label{exp_short}Running times of the text sampling algorithms in the case of small patterns ($2\leq m \leq 16$). The dashed red line represent the running time of the original Horspool algorithm.
Running times (in the $y$ axis) are represented in thousands of seconds. The $x$ axis represents the rank of the pivot character in the case of the new algorithm, while represents the number of removed characters in the case of previous algorithms.}
\end{figure}

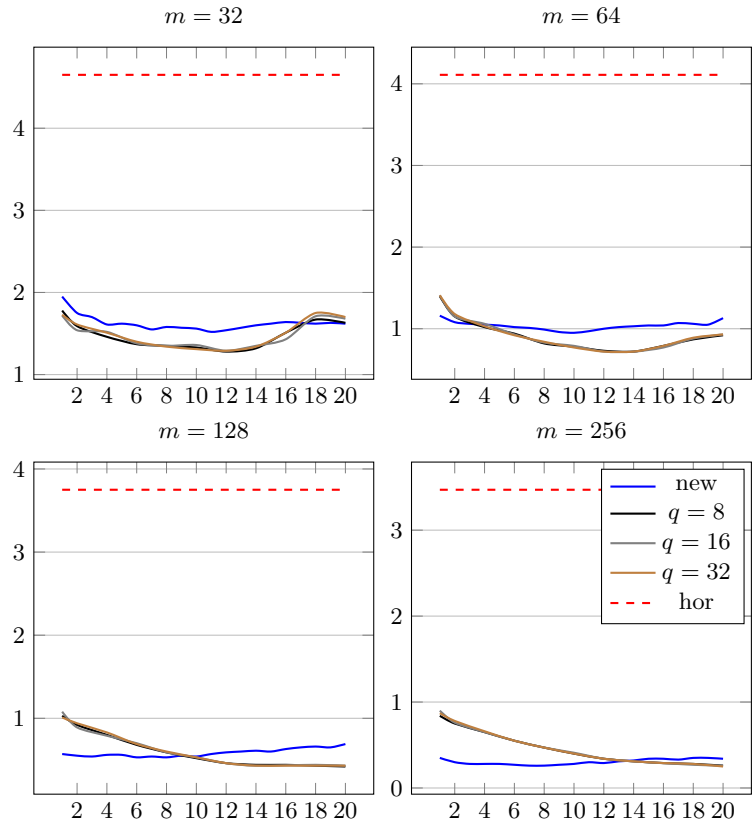
\begin{figure}[!t]
\begin{center}
\begin{tikzpicture}
	\begin{axis}[width={0.5\textwidth},height={6cm},	ymajorgrids = true,title={$m=32$},xtick={2,4,6,8,10,12,14,16,18,20}]
\addplot[thick,blue,smooth]  coordinates {
(1,1.95)	(2,1.75)	(3,1.7)	(4,1.61)	(5,1.62)	(6,1.60)	(7,1.55)	(8,1.58)	(9,1.57)	(10,1.56)	(11,1.52)	(12,1.54)	(13,1.57)	(14,1.6)	(15,1.62)	(16,1.64)	(17,1.63)	(18,1.62)	(19,1.63)	(20,1.62)};
\addplot[thick,black,smooth]  coordinates {
(1,1.78) (2,1.59)	(4,1.46)	(6,1.37)	(8,1.35)	(10,1.33)	(12,1.28)	(14,1.32)	(16,1.51)	(18,1.67)	(20,1.63)};
\addplot[thick,gray,smooth]  coordinates {
(1,1.72) (2,1.54)	(4,1.52)	(6,1.38)	(8,1.35)	(10,1.36)	(12,1.29)	(14,1.35)	(16,1.43)	(18,1.71)	(20,1.68)};
\addplot[thick,brown,smooth]  coordinates {
(1,1.73) (2,1.61)	(4,1.51)	(6,1.4)	(8,1.34)	(10,1.31)	(12,1.29)	(14,1.34)	(16,1.51)	(18,1.75)	(20,1.7)};
\addplot[thick,red, domain=1:20,dashed] ({\x},{4.65});
\end{axis}
\end{tikzpicture}
\begin{tikzpicture}
	\begin{axis}[width={0.5\textwidth},height={6cm},	ymajorgrids = true,title={$m=64$},xtick={2,4,6,8,10,12,14,16,18,20}]
\addplot[thick,blue,smooth]  coordinates {
(1,1.16)	(2,1.08)	(3,1.06)	(4,1.05)	(5,1.04)	(6,1.02)	(7,1.01)	(8,0.99)	(9,0.96)	(10,0.95)	(11,0.97)	(12,1)	(13,1.02)	(14,1.03)	(15,1.04)	(16,1.04)	(17,1.07)	(18,1.06)	(19,1.05)	(20,1.13)};
\addplot[thick,black,smooth]  coordinates {
(1,1.40) (2,1.15)	(4,1.02)	(6,0.94)	(8,0.82)	(10,0.78)	(12,0.73)	(14,0.72)	(16,0.79)	(18,0.87)	(20,0.92)};
\addplot[thick,gray,smooth]  coordinates {
(1,1.41) (2,1.15)	(4,1.06)	(6,0.93)	(8,0.83)	(10,0.79)	(12,0.72)	(14,0.72)	(16,0.77)	(18,0.88)	(20,0.93)};
\addplot[thick,brown,smooth]  coordinates {
(1,1.41) (2,1.18)	(4,1.03)	(6,0.92)	(8,0.84)	(10,0.77)	(12,0.72)	(14,0.72)	(16,0.79)	(18,0.89)	(20,0.93)};
\addplot[thick,red, domain=1:20,dashed] ({\x},{4.11});
\end{axis}
\end{tikzpicture}\\
\begin{tikzpicture}
	\begin{axis}[width={0.5\textwidth},height={6cm}, ymajorgrids = true,title={$m=128$},xtick={2,4,6,8,10,12,14,16,18,20}]
\addplot[thick,blue,smooth]  coordinates {
(1,0.57)	(2,0.55)	(3,0.54)	(4,0.56)	(5,0.56)	(6,0.53)	(7,0.54)	(8,0.53)	(9,0.55)	(10,0.54)	(11,0.57)	(12,0.59)	(13,0.6)	(14,0.61)	(15,0.6)	(16,0.63)	(17,0.65)	(18,0.66)	(19,0.65)	(20,0.69)};
\addplot[thick,black,smooth]  coordinates {
(1,1.03) (2,0.92)	(4,0.8)	(6,0.68)	(8,0.59)	(10,0.52)	(12,0.46)	(14,0.44)	(16,0.44)	(18,0.43)	(20,0.42)};
\addplot[thick,gray,smooth]  coordinates {
(1,1.08) (2,0.89)	(4,0.79)	(6,0.7)	(8,0.59)	(10,0.53)	(12,0.46)	(14,0.43)	(0.43)	(18,0.44)	(20,0.43)};
\addplot[thick,brown,smooth]  coordinates {
(1,1.01) (2,0.94)	(4,0.83)	(6,0.69)	(8,0.6)	(10,0.53)	(12,0.46)	(14,0.43)	(16,0.43)	(18,0.43)	(20,0.43)};
\addplot[thick,red, domain=1:20,dashed] ({\x},{3.75});
\end{axis}
\end{tikzpicture}
\begin{tikzpicture}
	\begin{axis}[width={0.5\textwidth},height={6cm},	ymajorgrids = true,title={$m=256$},xtick={2,4,6,8,10,12,14,16,18,20}]
\addplot[thick,blue,smooth]  coordinates {
(1,0.35)	(2,0.3)	(3,0.28)	(4,0.28)	(5,0.28)	(6,0.27)	(7,0.26)	(8,0.26)	(9,0.27)	(10,0.28)	(11,0.3)	(12,0.29)	(13,0.31)	(14,0.32)	(15,0.34)	(16,0.34)	(17,0.33)	(18,0.35)	(19,0.35)	(20,0.34)};
\addplot[thick,black,smooth]  coordinates {
(1,0.84) (2,0.75)	(4,0.65)	(6,0.55)	(8,0.47)	(10,0.4)	(12,0.34)	(14,0.31)	(16,0.29)	(18,0.28)	(20,0.26)};
\addplot[thick,gray,smooth]  coordinates {
(1,0.90) (2,0.76)	(4,0.65)	(6,0.55)	(8,0.47)	(10,0.41)	(12,0.34)	(14,0.31)	(0.29)	(18,0.27)	(20,0.26)};
\addplot[thick,brown,smooth]  coordinates {
(1,0.87) (2,0.78)	(4,0.66)	(6,0.55)	(8,0.47)	(10,0.4)	(12,0.34)	(14,0.31)	(16,0.29)	(18,0.28)	(20,0.25)};
\addplot[thick,red, domain=1:20,dashed] ({\x},{3.47});
\legend{new,$q=8$,$q=16$,$q=32$,hor}
\end{axis}
\end{tikzpicture}\\
\end{center}
\caption{\label{exp_long}Running times of the text sampling algorithms in the case of long patterns ($32\leq m \leq 256$). The dashed red line represent the running time of the original Horspool algorithm.
Running times (in the $y$ axis) are represented in thousands of seconds. The $x$ axis represents the rank of the pivot character in the case of the new algorithm, while represents the number of removed characters in the case of the \textsc{Ots} algorithms.}
\end{figure}

Figure \ref{exp_short} and Figure \ref{exp_long} show the searching times of the text sampling algorithms in the case of short patterns ($2\leq m \leq 16$) and long patterns ($32\leq m \leq 256$), respectively. The dashed red line represents the running time of the original Horspool algorithm.
Running times (in the $y$ axis) are represented in thousands of seconds. The $x$ axis represents the rank of the pivot character in the case of the new algorithm, while represents the number of removed characters in the case of the \textsc{Ots} algorithms.

In the case of short patterns best results are obtained by the new sampled approach selecting a small-rank pivot character. If compared with the original Horspool algorithm, the speed up obtained by the new approach goes from 32\% (for $m=2$) to 64\% (for m=16), while the gain obtained in comparison with the \textsc{Ots} algorithm decreases as the length of the pattern increases and going from 13\%, for $m=2$, to 7.7\%, in the case of $m=16$.

In the case of long patterns the difference between the running times of the two algorithms is negligible. However, if compared with the original Horspool algorithm, the speed up is much more evident and goes from 66\% (for $m=32$) to 91\% (for $m=256$).

\section{Conclusions and Future Works} \label{sec:conclusions}
We presented a new approach to the sampled string matching problem based on alphabet reduction and characters distance sampling. In our solution we divide the text in blocks of size $k$ and sample positions of such characters inside the blocks. During the searching phase the sampled data is used to filter candidate occurrences of the pattern. All the occurrences which are found during this first step are then verified for a whole match in the original text.

Our algorithm is faster than previous solutions in the case of short patterns and may require only 5\% of additional extra space. Despite this good performance we also proved that, when the underlying algorithm used for searching the sampled text for the sampled pattern achieves optimal worst-case and average-case time complexities, then also our algorithm attains the same optimal complexities, at least for patterns with a length of (at most) few hundreds of characters.

We applied our solution only to the case of a natural language text with a rather wide alphabet since the current approach doesn't work efficiently for small alphabets. It turns out indeed that the number of false positives located during the filtering phase increases as the size of the alphabet decreases. Thus it should be interesting to find a non-trivial strategy to extend this kind of solution also to the case of sequences over a small alphabet, like genome or protein sequences. We also wonder whether indexed solutions, as those based on the suffix tree, the suffix array and the FM-index of the text, should take advantage of the new text sampling approach. We intend to go in such directions in our future works.




\end{document}